\documentclass[journal,twoside,web]{ieeecolor}
\usepackage{generic}
\usepackage{amsmath,amssymb,amsfonts}
\usepackage{algorithmic}
\usepackage{graphicx}
\usepackage{textcomp}

\usepackage{standalone}
\usepackage{mathptmx} 
\usepackage{times} 
\usepackage{amsmath} 
\usepackage{amssymb}  
\usepackage{xfrac}
\usepackage[algo2e, norelsize, ruled]{algorithm2e}
\usepackage[short]{optidef}
\usepackage{booktabs}
\usepackage{microtype} %
\usepackage{units}
\usepackage{siunitx}
\usepackage{tabularx, tabulary}
\usepackage{caption}
\usepackage{subcaption}
\usepackage{paralist}
\usepackage[hidelinks]{hyperref}

\makeatletter 
\let\NAT@parse\undefined
\makeatother
\usepackage[isbn=false,backend=biber, style=ieee]{biblatex}
\newcommand{\ie}{i\/.\/e\/.,\/~}%
\newcommand{\eg}{e\/.\/g\/.,\/~}%
\newcommand{\cf}{cf\/.\/~}%
\newcommand{\T}{^{\mathsf{T}}}

\newtheorem{assum}{Assumption}

\newtheorem{thm}{Theorem}
\newtheorem{rem}{Remark}
\newtheorem{defn}{Definition}

\usepackage{tikz}

\def\BibTeX{{\rm B\kern-.05em{\sc i\kern-.025em b}\kern-.08em
    T\kern-.1667em\lower.7ex\hbox{E}\kern-.125emX}}
\begin{document}
\title{Parameter Filter-based Event-triggered Learning}
\author{Sebastian Schlor\authorrefmark{1}, Friedrich Solowjow\authorrefmark{1}, and Sebastian Trimpe 
	\thanks{\authorrefmark{1} Equally contributing}
\thanks{This work was supported in part by the Cyber Valley Initiative and the
	International Max Planck Research School for Intelligent Systems (IMPRS-IS).
Sebastian Schlor thanks the German Research Foundation (DFG) for support
of this work within grant AL 316/13-2 and within the German Excellence
Strategy under grant EXC-2075 - 285825138; 390740016. He acknowledges the support
by the Stuttgart Center for Simulation Science (SimTech)}
\thanks{Sebastian Schlor is with the University of Stuttgart, Institute for Systems Theory and Automatic Control, 70550 Stuttgart, Germany (e-mail: schlor@ist.uni-stuttgart.de). }
\thanks{Friedrich Solowjow and Sebastian Trimpe are with 
the Institute for Data Science in Mechanical Engineering, RWTH Aachen University, 52068 Aachen, Germany, and
the Intelligent Control Systems Group, Max Planck Institute for Intelligent Systems,
70569 Stuttgart, Germany  (e-mail: \{friedrich.solowjow, trimpe\}@dsme.rwth-aachen.de).}}
\maketitle

\begin{abstract}
Model-based algorithms are deeply rooted in modern control and systems theory. However, they usually come with a critical assumption -- access to an accurate model of the system. 
In practice, models are far from perfect. Even precisely tuned estimates of unknown parameters will deteriorate over time. 
Therefore, it is essential to detect the change 
to avoid suboptimal or even
dangerous behavior of a control system. 
We propose to combine statistical tests with dedicated parameter filters that track unknown system parameters from state data. 
These filters yield point estimates of the unknown parameters and, further, an inherent notion of uncertainty. 
When the point estimate leaves the confidence region, we trigger active learning experiments.
We update models only after enforcing a sufficiently small uncertainty in the filter.
Thus, models are only updated when necessary and statistically significant while ensuring guaranteed improvement, which we call event-triggered learning.
We validate the proposed method in numerical simulations of a DC motor in combination with model predictive control.
\end{abstract}


\section{INTRODUCTION}
\label{sec:introduction}

\begin{figure}[!htb]
	\includestandalone[width=\columnwidth, mode=buildnew]{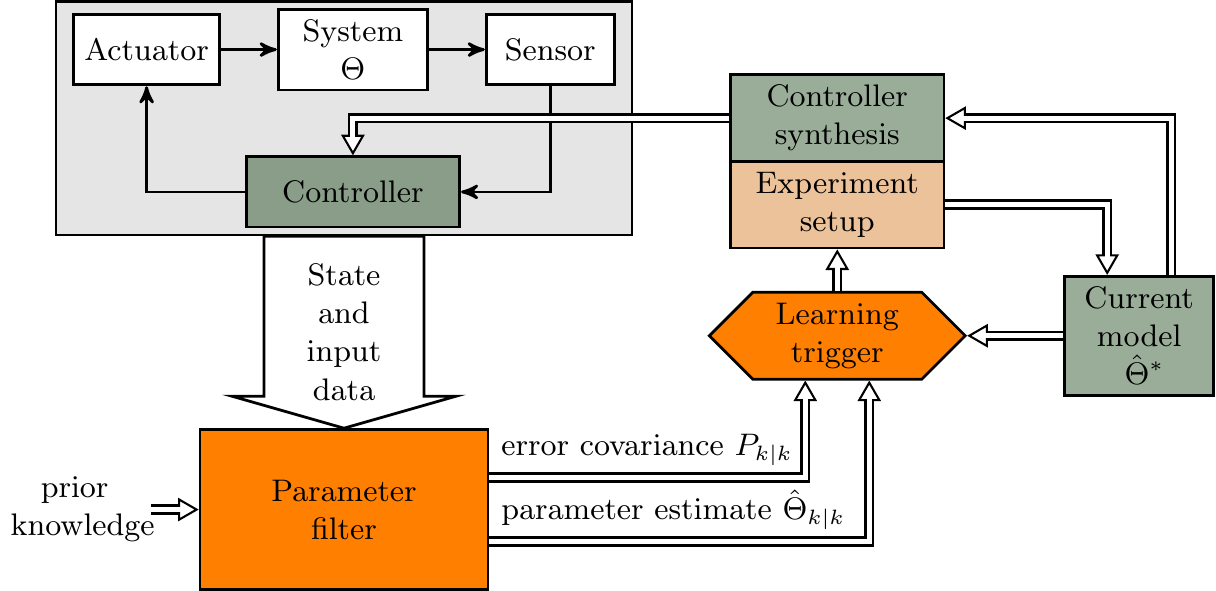}
	\caption{
		The figure schematically depicts the event-triggered learning framework.
		A model is used to synthesize a controller, which is applied to the system.
		The parameter filter outputs current estimates of model parameters with its uncertainty obtained from the measured data.
		The learning trigger compares the current model belief with the parameter filter estimate of the model. Updates are only triggered if there is a significant discrepancy between them.
	}
	\label{fig:tikz:BlockDiagramMethods}
\end{figure}
\IEEEPARstart{N}{ever} \emph{change a running system} is a popular heuristic to ensure the dependable operation of engineering systems.
At the same time, learning-based techniques are becoming increasingly popular in the control community, and learning usually requires some form of exploration to the system; that is, change. 
Despite the popularity, success stories are still rare. 
One key issue are the opposing objectives of many control tasks and the requirements for successful learning outcomes.
On the one hand, we usually aim for controllers that stabilize a system with the goal of avoiding deviations from a setpoint or reference. 
But, if a system is well regulated and barely moving, data is mainly dominated by noise.
Using this data for learning is a bad idea and leads to profound theoretical issues that might result in divergence of parameters and damage to the system.
Therefore, excitation is a critical requirement to learn something meaningful; however, often not desired as it deteriorates the control performance.  
Hence, learning permanently can be problematic. 
Instead, we want to detect the instances when updates are useful and -- only then -- effectively execute learning over a limited period of time.

The main idea of our event-triggered learning (ETL) framework is depicted in Figure~\ref{fig:tikz:BlockDiagramMethods}.
In the top left corner, we see the controlled system that, during normal operation, will have little excitation and thus yield uninformative data in most cases.
The parameter filter (bottom left) estimates the system's parameters to test if the model utilized for control is still reliable. 
So far, there is no learning involved.
Changes in the dynamics or environment might occur naturally and deteriorate the control performance or lead to constraint violations.
When there is an actual change in the system's behavior, the learning trigger shall detect this and triggers learning experiments.
During this learning phase, the system is actively excited to ensure accurate learning outcomes. 
After reaching sufficient accuracy, we stop the excitation, and the learning phase terminates. 
We update the model including the controller, and the parameter filter keeps monitoring the parameters.
Unless there is more change in the future, there is no more need for further learning. 


There is prior work on ETL that is closely tailored to specific applications and downstream tasks, such as linear quadratic control~\cite{Schluter2019} or the reduction of communication in networked control systems~\cite{Solowjow2020}. 
The triggers developed therein are based on a signal that is specific to the task: quadratic control performance \cite{Schluter2019} and communication rates \cite{Solowjow2020}. 
While this is helpful for the specific task, it also limits the algorithm to exactly that application.  
Here, we propose to build triggers based on the model quality. 
This way, we separate triggering from the downstream task and thus -- in principle -- make ETL applicable to any task that uses the model.  
Additionally, by considering the model quality as our target signal, we obtain statistical information regarding the parameters that is beneficial for designing the subsequent learning experiments. 

The proposed ETL approach yields the following three main improvements:
1) a statistical test that directly acts on the model quality;
2) point estimates that can be used as a new model; and
3) uncertainty quantification of the new model.
In contrast to previous ETL, 2) allows for efficient online learning, and 3) gives us a notion about whether the learned model is sufficient or further excitation is required.

The proposed design is flexible and can be combined with any model-based downstream algorithm or control architecture.
To demonstrate the flexibility of the proposed method, we consider a model predictive control task. 
Due to the optimization-based formulation of the controller architecture, we can readily incorporate learning objectives in the optimization, thus arriving at an active learning framework.
We show the benefits of the proposed method for a DC motor use case where we simulate and successfully detect changes in the dynamics.

\subsection{Related work}\label{sec:relwork}

The general idea of ETL, \ie learning only when necessary, has recently been proposed in \cite{so18,Solowjow2020,beuchert19,beuchert20,baumann2019event,Schluter2019,umlauft2019feedback,umlauft2020smart}, but in different settings than the one herein. Specifically, prior work on ETL builds triggers for a specific task or application.  
In particular, \cite{so18,Solowjow2020} provide the theory to reduce communication in networked control systems and was experimentally validated with IMU-sensor networks that monitor human gaits \cite{beuchert19,beuchert20}.
In \cite{baumann2019event}, ETL was extended to event-triggered pulse control and in~\cite{Schluter2019} to LQR control.
All of the above papers are not directly considering the learning problem and, further, do not leverage the statistical information that the learning trigger provides.
We discuss how the variance of our parameter filter can be reduced with appropriate control inputs and, very importantly, how it converges back to a steady-state during the nominal operation of the plant.
This steady-state variance is problematic for model updates and should be avoided due to the estimates performing a random walk in a possibly large uncertainty ellipsoid.
Purely adaptive methods that update models and controllers continuously often suffer from this issue, while our approach allows us to maintain accurate estimates since updates are connected to excitation and are only performed when the learning trigger detects a significant deviation.

Similar ideas that stress the importance of learning only when necessary have been, for example, introduced in~\cite{umlauft2019feedback,umlauft2020smart}. New training points are only added to a Gaussian process regression model when there is a significant improvement. 
Further, the main narrative of ETL is also consistent with insights from cognitive science, where human learning is quantified with internal models and surprisal \cite{butz2003anticipatory,humaidan2020fostering}.

Adaptive control algorithms (see~\cite{sastry2011adaptive} or~\cite{ioannou2012robust} for an overview) are closely related to ETL and consider a similar problem -- coping with partly unknown or changing systems.
While conceptually similar, at its core, it is different.  
Adaptive control strives to continuously update controllers; ETL does not and updates only sporadically when there is a need. 
It is well known that some adaptive control schemes suffer from temporary instability, so-called bursting~\cite{anderson2005failures}.
It can occur if the system converged to a steady-state and the measured signals are not persistently exciting.
Then, divergence may become a problem. 
In our work, we permanently estimate the system and its uncertainty, but we update only when needed.  This way, we avoid any bursting behavior.

Iterative learning control is similar as well but primarily developed for improving tracking performance in a repetitive setting. 
Instead of models, usually, the control input itself is optimized~\cite{ahn2007iterative,bristow2006survey}.
When looking closely at the setup, it is a different problem setting as in ETL, and to the best of our knowledge, changing dynamics are not considered.

Dual control circumvents divergence issues by considering two objectives simultaneously: optimizing control performance and 
minimizing model uncertainty.
A survey on dual control can be found in~\cite{unbehauen2000adaptive}, and a general textbook on adaptive dual control is given in~\cite{filatov2004adaptive}.
The authors of~\cite{Heirung2012} proposed a dual MPC scheme that takes the error covariance of the estimated model parameters into account in the cost function to be minimized. A similar control objective has already been introduced in \cite{wittenmark1975dual}.
In these approaches, however, model uncertainty is artificially increased over time to enforce the excitation of the system. Intuitively, the controller permanently forces the system to move in order to continuously update the model parameters. 
Our approach is different; we leverage statistical tests to detect the need for learning and only excite the system when necessary.


Detecting change in system dynamics is also a relevant topic in performance monitoring and assessment.
Closed-loop data is used to evaluate the control performance in terms of minimizing the output variance~\cite{Qin1998, harris1989assessment}.
The outcome of this analysis of variance (ANOVA) is then compared to a benchmark performance.
Often, minimum variance control serves as this benchmark; however, also user-specific criteria can define the desired performance if minimum variance control is not achievable nor desired~\cite{Jelali2006}.
In~\cite{Julien2004}, a performance benchmark for MPC is derived. Also, the benefit of updating the controller model is estimated. 
Further, model-plant mismatches are purely deduced on variance assessment. 
The above methods struggle with providing guidance for corrective action ~\cite{bauer2016current}.
We address these shortcomings by leveraging the inherent parameter filter uncertainty ellipsoid to derive optimal excitation signals.
These signals produce critical insights into the system's behavior for appropriate model updates.

There are many system identification techniques for estimating models of linear systems from data.
A broad overview is given in~\cite{Ljung1999}. 
The general idea of utilizing filtering techniques to estimate model parameters is also standard and has been investigated, \eg in~\cite{Goodwin1977,ljung1983theory}. 
We call these techniques parameter filters due to their functionality.
In~\cite{GARCIA201667}, a Kalman filter for state-space system identification is used in an iterative adaptive control scheme.
We consider a similar Kalman filter. 
However, in our work, the Kalman filter is used not only to estimate the system's parameters but also to derive trigger conditions.
Combining the filter with statistical tests and utilizing the beneficial posterior properties for learning on necessity is new and one of our key contributions.
Additionally, there has been a lot of research in designing optimal excitation signals to minimize the error covariance~\cite{Ljung1999,Bombois2011}.
We leverage the shape of the estimated uncertainty ellipsoid to design control inputs that minimize said ellipsoid.

\subsection{Contributions}

\begin{samepage}
	In summary, this article makes the following contributions:
	\begin{itemize}
		\item proposing a parameter filter-based ETL approach that yields point estimates and uncertainty ellipsoids that can directly be combined into a powerful statistical test;
        \item leveraging the induced uncertainty to decide if we need more data for suitable model updates and designing optimal input signals based on the point estimate; and
		\item incorporating robustness margins and prior knowledge into the statistical test to target specific parameters.
	\end{itemize}
\end{samepage}

\section{PROBLEM AND MAIN IDEA}\label{sec:problem}
Next, we introduce the considered system and control architecture.
This is followed by a precise problem formulation.

\subsection{System dynamics}
We consider a linear, discrete-time system
\begin{align}
	x_{k+1} = Ax_{k} + Bu_{k} + w_{k}
	\label{eq:LTIdef}
\end{align}
with the state $x_k\in \mathbb{R}^n$, the control input $u_k\in \mathbb{R}^m$, the random disturbance $w_k \in \mathbb{R}^n$  at discrete-time $k \in \mathbb{N}$, and the unknown system matrices $A \in \mathbb{R}^{n \times n}$ and $B \in \mathbb{R}^{n \times m}$ for some $n,m\in\mathbb{N}$.
Furthermore, the disturbances $\{w_k\}$ are assumed to be independent and identically distributed (i\/.\/i\/.\/d\/.\/) and drawn from a normal distribution with mean $0$ and known covariance $\Sigma_{w}$, which is denoted by $w_k\sim\mathcal{N}(0,\Sigma_{w})$.
Thus, the unknown parameters of the system can be fully parameterized by the stacked matrix $\Theta\T = \begin{bmatrix}A ~ B\end{bmatrix} \in \mathbb{R}^{n \times (n+m)}$.
Since the true system parameters $\Theta$ are unknown, we consider a model of the system, which is denoted by $\hat{\Theta}^* = \begin{bmatrix}\hat{A}^* ~ \hat{B}^*\end{bmatrix}$.
Here, this model will be used to design a controller (cf. controller synthesis block in Figure~\ref{fig:tikz:BlockDiagramMethods}).
Further, assume that the states and inputs are subject to constraints $x_k \in \mathcal{X} \subseteq \mathbb{R}^n$ and $u_k\in \mathcal{U} \subseteq \mathbb{R}^m$ for all $k \in \mathbb{N}$.

To ensure a well-behaved control algorithm, we require the following mild assumptions.
\begin{assum}
	\begin{enumerate}
		\item The pair $(A,B)$ is stabilizable.
		\item The sets $\mathcal{X}$ and $\mathcal{U}$ are compact, convex and contain the origin.
	\end{enumerate}
\end{assum}
The standard assumption of stabilizability ensures that the states of the system remain bounded under suitable control.
The second assumption states that the equilibrium of the system is located in the admissible state space.

\subsection{Objective: ETL with recursive learning}

We assume that the parameters $\Theta$ for system \eqref{eq:LTIdef} can change. 
Further, changes are assumed to be rare and sudden. 
The main problem we consider here is detecting these changes and subsequently updating the static model $\hat{\Theta}^*$ accordingly.

Clearly, we should leverage all of the available information to be as efficient as possible.
Despite a change in the dynamics, the old model still contains valuable information that can be leveraged to relearn the model $\hat{\Theta}^*$ efficiently. 

\subsection{Main idea: ETL architecture}
\label{sec:ETLarchi}

We propose to utilize a parameter filter (Sec.~\ref{sec:sysid}) that outputs a point estimate $\hat{\Theta}_{k|k}$ of the unknown system parameters $\Theta$ and additionally yields a posterior variance $P_{k|k}$.
Usually, a Kalman filter is used to estimate unknown states from observations.
Here, we estimate unknown system parameters from data.
Due to the underlying normal distributions in the filter, it is possible to design the learning trigger (Sec.~\ref{sec:KFTrigger}) based on well-established statistical tests for the null hypothesis 
\begin{equation}
 H_0: \hat{\Theta}_{k|k}=\hat{\Theta}^*.
 \label{eq:null}
\end{equation}
Whenever we reject the hypothesis, we trigger learning to adapt to the changes in the dynamics.

The filter output $\hat{\Theta}_{k|k}$ is usually subject to a high variance $P_{k|k}$ when there is only little movement in the system. 
Therefore, it is problematic to use it for directly updating the model $\hat{\Theta}^*$.
Instead, we first shrink the ellipsoid through optimized control inputs (Sec.~\ref{sec:Experiment}). 
This allows us to ensure that we relearn accurate models $\hat{\Theta}^*$.
Due to the recursive nature of the filter, we also obtain an effective way to incorporate prior knowledge. 
The interactions between the components are depicted in Figure~\ref{fig:tikz:BlockDiagramMethods}.





\section{RECURSIVE SYSTEM IDENTIFICATION}
\label{sec:sysid}

Identifying linear systems is a classical problem that has been addressed in many textbooks, \eg \cite{Ljung1999} and~\cite{ljung1983theory}, specifically for recursive system identification. Nonetheless, the identification of linear systems is still subject of recent research, \eg \cite{simchowitz2018learning,simchowitz2019learning}. 
In our work, we discuss the question of \emph{when to learn} by leveraging a state-space formulation of the parameter filter, which we have not found elsewhere. 

We start by summarizing recursive system identification approaches that are primarily based on~\cite{Goodwin1977}. 
These estimators provide point estimates and corresponding uncertainty ellipsoids. 
Further, depending on often implicitly made assumptions, the estimators have different properties.
At the same time, these differences are critical when deciding what filter to use for an ETL approach. 
Therefore, we first present the most common approaches to recursive identification -- A) least squares and B) Kalman filter-type estimators. 
Afterward, we show why the latter is most beneficial for the problem at hand.

For least squares-type approaches, we usually assume time-invariant dynamics as introduced in \eqref{eq:LTIdef}. We show how to estimate the static model parameters with standard techniques and afterward consider addressing potential changes through forgetting.
For example, we can reduce the influence of old data points through appropriate scaling. 

Alternatively, the change can explicitly be taken into account and directly encoded into the estimation procedure through appropriate assumptions.
Mathematically, we obtain the structure
\begin{align}\label{eq:LTVSystem}
	x_{k+1} &= A_k x_k + B_k u_k +w_k
\end{align}
and aim for time-varying estimates $\hat{A}_k$ and $\hat{B}_k$ of the nominal parameters $A_k$ and $B_k$. This approach leads to a Kalman filter-style parameter filter.

\subsubsection*{Vectorized process model}

To ease and unify notation, we introduce a vectorized system, \ie the model parameters will be contained in a common vector.
The general system~\eqref{eq:LTVSystem} can be rewritten as
\begin{align} \label{eq:LTVSystemTheta}
	\begin{split}
		x_{k+1}\T &= d_k\T \Theta_k + w_k\T\,,
	\end{split}
\end{align}
where  $\Theta_k\T = \begin{bmatrix}A_k ~ B_k\end{bmatrix} \in \mathbb{R}^{n \times (n+m)}$ is the stacked matrix containing all the system parameters and $d_k\T = \begin{bmatrix}x_{k}\T ~ u_{k}\T\end{bmatrix} \in \mathbb{R}^{n+m}$ contains the state and input at time $k$.

Further, we can vectorize the system matrix using the vectorization operator $\mathrm{vec}(\cdot)$, which stacks the columns of a matrix to one vector. We write $z_k = \operatorname{vec}(\Theta_{k}) \in \mathbb{R}^{n(n+m)}$.
The linear system \eqref{eq:LTVSystemTheta} can now be written in vector form as
\begin{align} 
		x_{k+1} &= C_k z_k + w_k\,,
	\label{eq:vecSysDescription}
\end{align}
with $C_k = \mathrm{I}_n \otimes d_k\T \in \mathbb{R}^{n \times n(n+m)}$, and $\otimes$ the Kronecker product.

\subsection{Least squares estimators}\label{subsec:LSEStimation}
The objective of least squares parameter estimation is to find estimated model parameters $\hat{z}_k$, which minimize the squared
prediction error at the current time as
\begin{align}\label{eq:minPredError}
	\hat{z}_k = \arg\underset{\hat{z}}{\min} ~\mathbb{E}\left[ (x_{k+1}- C_k\hat{z})\T (x_{k+1}- C_k\hat{z}) \right].
\end{align}
Clearly, the estimator highly depends on the available data $C_k$ and, additionally, on the (often implicitly made) assumptions on the parameters $z_k$.
Let
\begin{gather}
	X = \begin{bmatrix}
		x_1\\
		\vdots\\
		x_{k}
	\end{bmatrix},\qquad
	C = \begin{bmatrix}
		C_0\\
		\vdots\\
		C_{k-1}
	\end{bmatrix}
\end{gather}
be the stacked data vectors and matrices generated by the system from time $0$ to $k$. 
First, let us consider the standard time-invariant problem, \ie
\begin{align}\label{eq:constZ}
	z_{k+1} &= z_k \quad \forall k.
\end{align}
Then, the analytic solution to~\eqref{eq:minPredError} is given by 
\begin{align}\label{eq:ls_sol}
	\hat{z}_k = \left(C\T C\right)^{-1} C\T  X\,.
\end{align}
This batch least squares estimate can be reformulated recursively as
\begin{align}
	K_k &= P_{k-1} C_{k-1}\T \left(\mathrm{I} + C_{k-1} P_{k-1} C_{k-1}\T\right)^{-1}\\
	\hat{z}_k &= \hat{z}_{k-1} + K_k (x_k - C_{k-1}  \hat{z}_{k-1})\\
	P_k &= (\mathrm{I} - K_k C_{k-1}) P_{k-1}\,
\end{align}
with the starting values $\hat{z}_0$ and $P_0$ (see \eg \cite{Goodwin1977}).

In this estimator, the assumption of a constant model leads to a monotonically decreasing error covariance matrix~$P_k$.
With increasing amounts of data, the estimate converges, and the gain matrix $K_k$ decreases with respect to a suitable norm.
Thus, new measurements only have a small impact on the estimated parameters. The recursive least squares estimator eventually loses its adaptivity and stops updating.
It weights all samples equally and does not give preference to more recent data.
While this is desirable for constant model parameters, there will be issues if the system changes over time.

As stated earlier, one approach to making recursive least squares adaptable to changing parameters is the recursive least squares estimator with exponential forgetting~\cite{Goodwin1977}.
Here, observations that are $l$ time steps before the current time $k$ are weighted by $\lambda^{k-l}$, with the constant $\lambda\in (0,1)$ (typically $\lambda = 0.9, \dots, 0.99$).
Thus, increasing the influence of recent samples on the estimator.
This is equivalent to adapting~\eqref{eq:vecSysDescription} and~\eqref{eq:constZ} with a weighted output equation
\begin{align}
	\lambda^{k-l} x_{l+1} &= \lambda^{k-l} C_l z_l + w_l   	  & w_l \sim \mathcal{N}\left(0,\Sigma_w\right)\,.  \label{eq:optAss:RLSl1}
\end{align}
The scaling with $\lambda^{k-l}<1$ decreases the absolute value of the state. 
Hence, the squared prediction error is directly effected, and the influence of the corresponding point is lowered.
This can be interpreted as artificially decreasing the signal-to-noise ratio for old data since the additional disturbance $w_k$ is not weighted.
Even though we keep $z$ constant, the influence of past data is reduced because data with a lower signal-to-noise ratio is less informative.
As a consequence, the underlying model \eqref{eq:optAss:RLSl1} is equivalent to an unweighted observation with disturbances $w_k$ whose covariance matrix is time-varying.
We can simply divide both sides of \eqref{eq:optAss:RLSl1} by $\lambda^{k-l}$ and substitute $\left(\frac{1}{\lambda}\right)^{k-l} w_l$ by $w_l$.
Then, we obtain for all times $l \le k$
\begin{align}
	\begin{aligned}
		z_{l+1} &= z_l 																						&\qquad &\\
		x_{l+1} &= C_l z_l + w_l   	  &\qquad & w_l \sim \mathcal{N}\left(0,\left(\frac{1}{\lambda^2}\right)^{k-l}  \Sigma_w\right)
	\end{aligned}
\end{align}
as the data generation model for which recursive least squares with exponential forgetting yields the optimal estimate.

\subsection{Kalman filter approach}\label{subsec:KF}
In order to explicitly address potential change in the dynamics (\cf \eqref{eq:vecSysDescription}), we assume the additional structure for the changes
\begin{align}
	\begin{aligned}
		z_{k+1} &= z_k + \Delta z_k 												  					  &\qquad &\Delta z_k \sim \mathcal{N}\left(0,\Sigma_z\right)\\
		x_{k+1} &= C_k z_k + w_k	   	 	&\qquad & w_k \sim \mathcal{N}\left(0,\Sigma_w\right)\,,  \label{eq:optAss:KF}
	\end{aligned}
\end{align}
where the random variable $\Delta z_k$ induces the change in the system parameters. 
By $\Sigma_{z}$ we denote the covariance of the assumed model changes $ \Delta z_k$, which can be used as a tuning parameter. Here, one may also incorporate prior knowledge of the possible system changes, \eg which parameters can be affected in case of a load change, etc. 
We further consider the following standard assumptions. 
\begin{assum}\label{ass:KF}
	\begin{enumerate}
		\item The disturbance sequences $\{\Delta z_k\}$ and $\{w_k\}$ are i\/.\/i\/.\/d\/.\/ and $\mathbb{E}\left(\Delta z_k w_k\T\right) = 0$.
		\item The initial value $z_{0}$ is independent of $\Delta z_k$ and $w_k$.
	\end{enumerate}
\end{assum}

For such a system, the Kalman filter is the optimal state estimator~\cite{Matisko2012}.
In particular, it is the optimal Bayesian estimator that keeps track of the full posterior distribution.

Typically, the Kalman filter is used to estimate the state vector $\hat{x}_k$ of a linear state-space model when only the observations $y_k$ are available, and the state $x_k$ is hidden.
In our case, we assume access to the whole state of the system, but the systems transition matrices $A_k$ and $B_k$, respectively $z_k$, are unknown and might change over time. Our goal is a filter that yields the system parameters.

Essentially, we lift the standard estimation problem one level higher to estimate the process model from given states instead of estimating the state from given measurements.
Thus, for this vectorized system, the Kalman filter is given by
\begin{align}
	\hat{z}_{k+1|k} &= \hat{z}_{k|k}\\
	P_{k+1|k} &= P_{k|k} + \Sigma_{z}\\
	e_{k+1} &= x_{k+1} - C_k \hat{z}_{k+1|k}\\
	S_{k+1} &= C_k P_{k+1|k} C_k\T + \Sigma_w\\
	K_{k+1} &= P_{k+1|k} C_k\T S_{k+1}^{-1}\\
	\hat{z}_{k+1|k+1} &= \hat{z}_{k+1|k} + K_{k+1} e_{k+1}\\
	P_{k+1|k+1} &=P_{k+1|k} - K_{k+1} C_k P_{k+1|k} \,\label{eq:PUpdate}
\end{align}
with initial values $P_{0|0}$ and $\hat{z}_{0|0}$.
In this notation, a subscript ${l|k}$ indicates the estimate for time step $l$ given the data up to time step $k$.
If $\Sigma_{z}$ is chosen \emph{large}, the Kalman filter does weigh more recent samples higher than older ones. 
The resulting scheme is a recursive estimation method and allows for online learning.

Due to~\eqref{eq:optAss:KF} and Assumptions \ref{ass:KF},
the Kalman filter estimate $\hat{z}_{k|k}$ is normally distributed with mean $\mathbb{E}(z_k) $ and error covariance $P_{k|k} = \mathbb{E}( (\hat{z}_{k|k} - z_k) (\hat{z}_{k|k} - z_k)\T )$ (\cf \cite{Spall1984a}).

\section{PARAMETER FILTER LEARNING TRIGGER} \label{sec:KFTrigger}

In this section, we design the learning trigger, which detects significant deviations between models and dynamical systems (\cf Figure~\ref{fig:tikz:BlockDiagramMethods}). 
We derive the distribution of the parameter filter under the assumption that there is no change in dynamics ($\Delta z_k = 0$). Afterward, we propose a statistical test that validates if the data is consistent with the assumption $\Delta z_k = 0$.
This is exactly the same hypothesis as \eqref{eq:null} stated in Section~\ref{sec:ETLarchi}. In the following, we use the notation $\hat{z}$ instead of $\hat{\Theta}$ since we consider the equivalent vectorized formulation for the tests.

\subsection{Main idea of the learning trigger}

First, we explain the main idea of the learning trigger.
Essentially, there are three critical parts: i) the current model believe $\hat{z}^*$, which is kept constant and utilized for potential down-stream tasks, ii) the point estimate of the parameter filter $\hat{z}_{k|k}$, and iii) the covariance ellipsoid $P_{k|k}$.
Of course, there is also the ground truth $z_k$; however, these parameters are unknown.

Therefore, we construct a statistical test around the objects i) - iii) to infer if the model believe $\hat{z}^*$ is significantly different from the ground truth $z_k$. In particular, we can guarantee that the ground truth $z_k$ is with high probability contained inside an ellipsoid around the current estimate $\hat{z}_{k|k}$. 
Hence, when the old model belief $\hat{z}^*$ shows a large deviation from $\hat{z}_{k|k}$, we can also infer that the old model is not consistent anymore with the ground truth. Next, we make the main idea of the learning trigger mathematically precise and derive the corresponding distributions and confidence ellipsoids.

\begin{figure}[tb]
	\includestandalone[width=\columnwidth, mode=buildnew]{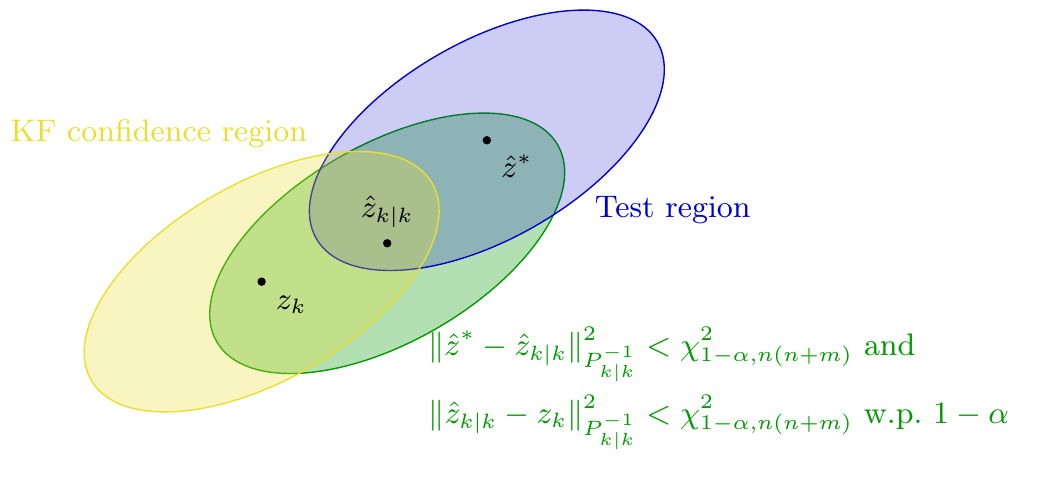}
	\caption[Visualization of confidence regions of the Kalman filter parameter estimation and model test regions.]{Visualization of confidence regions of the Kalman filter (KF) parameter estimation and model test regions.
		The estimated parameters $\hat{z}_{k|k}$, the true system parameters $z_k$, and the model parameters $\hat{z}^*$ are depicted as points in parameter space. The yellow ellipsoid is the confidence region of the Kalman filter estimation in which the error between $z_k$  and $\hat{z}_{k|k}$ is contained with probability $1-\alpha$.
		If the parameter estimate is inside the test region around the current model, both, the model and the true system parameters, are located inside the green ellipsoid with probability $1-\alpha$.}
	\label{fig:Ellipses}
\end{figure}

\subsection{Learning trigger}
We start with well-known connections between normal and $\chi^2$-distributions~\cite{corder2014nonparametric}.
In particular, due to the property  $\left(\hat{z}_{k|k} -z_k\right) \sim \mathcal{N}(0,P_{k|k})$ of the Kalman filter estimate, we know that $\left(\hat{z}_{k|k} - z_k \right) P_{k|k}^{-1} \left(\hat{z}_{k|k} - z_k \right)\T < \chi^2_{1-\alpha,n(n+m)}$ with probability $1-\alpha$ for some $\alpha\in\left(0,1\right)$.
By $\chi^2_{1-\alpha,n(n+m)}$ we denote the $1-\alpha$ quantile of the chi-square distribution with $n(n+m)$ degrees of freedom.
This corresponds to the confidence region in Figure~\ref{fig:Ellipses} depicted in yellow.
The estimated parameters $\hat{z}_{k|k}$ are located inside this ellipsoid around the true system $z_k$ with probability $1-\alpha$.
As a short notation for an expression $z P^{-1} z\T$ we use $\| z\|^2_{P^{-1}}$ in the figure since $P$ as a covariance is positive definite.
If the null-hypothesis \eqref{eq:null} is true, we obtain $\left(\hat{z}_{k|k} - \hat{z}^*\right) P_{k|k}^{-1} \left(\hat{z}_{k|k} - \hat{z}^*\right)\T < \chi^2_{1-\alpha,n(n+m)}$ with probability $1-\alpha$.

Based on these properties, we propose the learning trigger
\begin{align}
		\gamma_{\mathrm{learn}} = 1  \iff
		\left(\hat{z}_{k|k} - \hat{z}^*\right) P_{k|k}^{-1} \left(\hat{z}_{k|k} - \hat{z}^*\right)\T > \chi^2_{1-\alpha,n(n+m)}
		\label{eq:KFTest}
	\end{align}
as a statistical test of level $\alpha$.
Thus, a trigger event occurs when the data-based estimate $\hat{z}_{k|k}$ and the fixed model $\hat{z}^*$ are most likely to differ more than what is expected by the uncertainty in the data.
Further, by using this learning trigger, the probability of false trigger events, \ie triggering despite having an accurate model, is bounded as stated in the following theorem.

\begin{thm}
Consider the dynamical system \eqref{eq:optAss:KF} and let the Assumptions~\ref{ass:KF} hold.
	If the learning trigger \eqref{eq:KFTest}
	is used, then in case of a perfect model, \ie $\hat{z}^* = z_k$,
    \begin{align}
		\mathbb{P} \left[\gamma_{\mathrm{learn}} = 1\right] \le \alpha\,.
	\end{align}
	\label{thm:KFTrigger}
\end{thm}
\begin{proof}
The theorem follows directly from the properties of the Kalman parameter filter.
By design, we ensure that the estimates are subject to a normal distribution.
Further, through direct access to the variance, we can normalize the distribution and apply a standard $\chi^2$ test. The test then directly induces the confidence region.
\end{proof}

Thus, the risk of false trigger events and at the same time the sensitivity of the trigger can be adjusted by the probability $\alpha$.
If \eqref{eq:KFTest} is satisfied, the null hypothesis is rejected, and a new model needs to be identified.

The test statistic is also known as squared Mahalanobis distance, which is a generalization of the Euclidean distance to multivariate spaces with different scaling~\cite{DeMaesschalck2000}.
If the covariance matrix is non-singular, the Mahalanobis distance fulfills all properties of a metric.
It scales the space according to the covariance such that points lying on equal contour lines of the normal distribution have the same distance from the center point.
It can also be viewed as a distance of a point to a normal distribution.
This is also how our test can be interpreted. 
The Kalman filter gives a conditional distribution of parameter values with mean $\hat{z}_{k|k}$ and error covariance $P_{k|k}$.
Thus, if the assumptions are fulfilled, $\hat{z}_{k|k}$ to $z_k$ has squared Mahalanobis distance $\| \hat{z}_{k|k} - z_k \|^2_{P_{k|k}^{-1}} < \chi^2_{1-\alpha,n(n+m)}$ with probability $1-\alpha$.

If the test yields $\| \hat{z}_{k|k} - \hat{z}^* \|^2_{P_{k|k}^{-1}} < \chi^2_{1-\alpha,n(n+m)}$, we can conclude from the triangle inequality that at least with probability $1-\alpha$ the model and the true system parameters fulfill $\| \hat{z}^* - z_k \|^2_{P_{k|k}^{-1}} < \sqrt{2} \chi^2_{1-\alpha,n(n+m)}$. 
In Figure~\ref{fig:Ellipses}, this can be interpreted as we test if $ \hat{z}_{k|k}$ is inside the blue ellipsoid \emph{test region} around the model $\hat{z}^*$.
From the Kalman filter, we get the yellow \emph{confidence region}, which contains $\hat{z}_{k|k}$ with probability $1-\alpha$.
The boundary of the ellipsoids has squared Mahalanobis distance $\chi^2_{1-\alpha,n(n+m)}$ from the corresponding center points.
Since $\hat{z}_{k|k}$ is contained in both regions, the confidence region and the test region must intersect.
Hence, the squared distance between $z_k$ and $\hat{z}^*$ fulfills $\| \hat{z}^* - z_k \|^2_{P_{k|k}^{-1}} < \sqrt{2} \chi^2_{1-\alpha,n(n+m)}$ with probability $1-\alpha$.

\begin{rem}
	If the assumptions on the noise terms are not fulfilled, the normal distribution of the Kalman filter estimate does not follow immediately; however, there are properties that can ensure asymptotic convergence to a normal distribution~\cite{Spall1984}.
	Noteworthy, the parameter estimate is still asymptotically normal distributed without the assumption of Gaussian disturbances if it holds that the disturbances have zero mean, \ie $\mathbb{E}(\Delta z_k) = 0,~ \mathbb{E}(w_k) = 0$, and the disturbance covariance matrices and the initial error covariance are bounded.
\end{rem}
\begin{rem}
In Theorem~\ref{thm:KFTrigger}, we assume that the additive process disturbance $w_k$ is independent of the state and has constant covariance~$\Sigma_{w}$.
In practice, however, these assumptions could be violated if the controlled system is slightly nonlinear, for example.
One approach to increase the robustness of the algorithm against such deviations is to assume worse disturbances in the model than those that are actually observed.
Due to the over-approximation of the process noise, the test region of the Kalman filter trigger increases.
This leads to a much wider margin between the test statistics and the threshold value compared to tests with ideal noise assumptions.
This results in a lower sensitivity to small system changes but also in higher robustness to violated model assumptions.
\end{rem}

\section{EXPERIMENT DESIGN }\label{sec:Experiment}

The error covariance $P_{k|k}$ is an essential object of the proposed learning trigger and should be small for an effective test.
At the same time, it is a measure of the accuracy of the corresponding parameter estimate.
Thus, after a trigger event occurs, a learning experiment should be carried out to reduce the uncertainty of the estimate.
In general, the accuracy of a model estimate highly depends on the data and, in particular, the excitation of the system.
For example, consider the least squares estimator \eqref{eq:ls_sol} from Section~\ref{sec:sysid}.
Here, the matrix $C\T C$ can be written as $C\T C = \mathrm{I}_n \otimes \left(\sum_{t=1}^{k-1} d_t d_t\T\right)$.
The collected data must contain at least $(n+m)$ linearly independent vectors $d_t$ such that the inverse of the matrix $C\T C$ exists.
The rank condition on $\sum_{t=1}^{k-1} d_t d_t\T$ is closely related to persistency of excitation of the data sequence $\{d_t\}$.

\subsection{Persistency of excitation and observability}
Persistency of excitation has been introduced in several slightly different notations (cf. \cite{JOHNSTONE1982,Bai1985,Caines1984,GREEN1986351}). In a nutshell, if data has a sufficiently rich information content and gives insight into the system dynamics, then it is possible to guarantee convergence of statistical estimators. 

For finite sequences, persistency of excitation over an interval is defined in~\cite{GREEN1986351}.
\begin{defn}[Persistency of excitation]
	A sequence of data $\{d_t\}_{t=0}^{k-1}$, $d_t\in \mathbb{R}^{n+m}$ is persistently exciting over the time interval $\left\{0,\ldots,k-1\right\}$, if there exists $\epsilon>0$ such that
	\begin{align}
		\sum_{t=0}^{k-1} d_t d_t\T  \succeq  \epsilon\mathrm{I}_{n+m}.
	\end{align}
\end{defn}

Here, we make use of the Loewner order for positive semi-definite matrices. For two positive semi-definite matrices $M_1$ and $M_2$, we say that $M_1 \succeq M_2$ if $M_1-M_2$ is positive semi-definite.
Thus, if $\{d_t\}_{t=0}^{k-1}$ is persistently exciting, the matrix $C\T C$ is positive definite and has an inverse.

For $\{d_t\}_{t=0}^{k-1}$ to be persistently exciting, both components $\{x_t\}_{t=0}^{k-1}$ and $\{u_t\}_{t=0}^{k-1}$ must be persistently exciting.
Since the input sequence $\{u_t\}_{t=0}^{k-1}$ is a design variable, one purpose of experiment design is to ensure persistency of excitation of $\{x_t\}_{t=0}^{k-1}$.
This is contrary to the control objectives during normal operation, where deviations from the setpoint are suppressed. 

Persistency of excitation of $\{d_t\}_{t=0}^{k-1}$ induces observability of the lifted parameter system~\eqref{eq:optAss:KF}.
A linear discrete-time parameter varying system~\eqref{eq:optAss:KF} is said to be completely observable if the observability matrix
\begin{align}
	O_{0} = \begin{bmatrix}
		C_{0} ~
		C_{1} ~
		\ldots ~
		C_{n(n+m)}
	\end{bmatrix}\T
\end{align}
has $\operatorname{rank}(O_{0}) = n(n+m)$~\cite{witczak2017necessary}.
This is the same condition as for persistency of excitation.
Thus, without disturbances, the parameter vector $z_k$ could be determined from persistently exciting data.

While persistency of excitation ensures well-defined solutions to the estimation problem, it is not sufficient to control the estimation error.
Indeed, this means that the estimate may diverge under permanent updates.
Next, we consider an active learning problem to reduce the posterior variance of the estimator through optimized excitation of the system.

\subsection{Active learning}
The error covariance $P$ of the estimator is a natural measure of the accuracy of the estimate. 
Here, we aim at minimizing the trace of $P$, which corresponds to minimizing the sum of the eigenvalues of the error covariance.
If $\operatorname{trace}(P)$ is minimized by the data of the experiment, the data acquisition is called A-optimal~\cite{Pronzato2008}.
By doing such an A-optimal experiment, the sum of squared lengths of principal semi-axes of the confidence ellipsoid is minimized.
For the Kalman filter, the mean and covariance are usually obtained recursively. 
However, there exist also closed forms as derived in~\cite{Han2010}, where the influence of the experiment design on the covariance can be investigated.


The excitation during the experiment could be provided by external reference signals such as pseudo-random binary signals or chirp signals.
However, then also input and state constraints need to be considered.
A different approach was presented in~\cite{Heirung2012}. 
There, a dual MPC is proposed, which considers the expected trace of the error covariance matrix of a recursive least squares estimator with exponential forgetting in the optimization problem.
Here, this approach is adopted and modified such that the MPC propagates a model of the Kalman filter instead of the recursive least squares estimator. 
Its error covariance is jointly minimized in the optimization problem.
Thus, then we obtain the objective function:
\begin{samepage}
	\begin{subequations}
		\label{eq:expMPC}
		\begin{alignat}{2}
			&\hspace{-7mm} \!\min_{X_{t_j},U_{t_j}}        &\qquad& \sum_{k = 0}^{N-1} x_{k|t_j}\T Q x_{k|t_j} + u_{k|t_j}\T R u_{k|t_j} + x_{N|t_j}\T Q_N x_{N|t_j} \notag \\
			&												&\qquad	&\qquad + \nu\operatorname{trace}(\widetilde{P}_{k+1|k+1}) \tag{\vspace{-0.5\baselineskip}\ref{eq:expMPC}}
		\end{alignat}%
	\vspace{-7mm} 
		\begin{alignat*}{4}
			&\text{s.t.} &  && x_{0|t_j} &= x_{t_j},\\
			&                  &&& x_{k+1|t_j} &= \hat{A}x_{k|t_j} + \hat{B}u_{k|t_j}, &~ k \in \{0,\dots,N-1\}\\
			&                  &&& x_{k|t_j} &\in \mathcal{X}, &\quad k \in \{0,\dots,N-1\}\\
			&                  &&& u_{k|t_j} &\in \mathcal{U}, &\quad k \in \{0,\dots,N-1\}\\
			&                  &&& x_{N|t_j} &\in \mathcal{X}_N, &\\
			&                  &&& \widetilde{P}_{0|0} &= P_{t_j|t_j}, &\\
			&                  &&& \widetilde{C}_{k} &=\mathrm{I}_n \otimes \begin{bmatrix}
				x_{k|t_j}\T & u_{k|t_j}\T
			\end{bmatrix}, & k \in \{0,\dots,N-1\}\\
			&                  &&& \widetilde{S}_{k+1} &=\widetilde{C}_k \widetilde{P}_{k+1|k} \widetilde{C}_k\T + \Sigma_w, &k \in \{0,\dots,N-1\}\\
			&                  &&& \widetilde{K}_{k+1} &=\widetilde{P}_{k+1|k} \widetilde{C}_k\T \widetilde{S}_{k+1}^{-1}, &k \in \{0,\dots,N-1\}\\
			&                  &&& \widetilde{P}_{k+1|k+1} &=\widetilde{P}_{k+1|k} - \widetilde{K}_{k+1} \widetilde{C}_k \widetilde{P}_{k+1|k}, &k \in \{0,\dots,N-1\}&.
		\end{alignat*}
	\end{subequations}
\end{samepage}%
Here, $\nu$ is a weighting parameter. Our proposed method is summarized and abstracted in Algorithm~\ref{algo}.

Unfortunately, $\widetilde{P}_{k|k}$ is a nonlinear function of the system's states and inputs.
Therefore, the optimization problem becomes nonlinear and nonconvex.
To make the computation feasible, the standard MPC problem without the cost of the covariance matrix can be computed and used as an initial guess for the harder optimization problem.
Then, even if no global optimum is found, a feasible solution can be provided.
In~\cite{Heirung2012}, no stability guarantees are given about the considered approach.
In the case of re-identification after a wrong model is detected, this is a hard problem since the model used for control itself is not trustworthy.
In~\cite{Anderson2018}, a robust re-identification procedure for MPC is given.
However, there the considered system changes are bounded such that the controller is stabilizing for all possible system changes.

\section{SIMULATION EXAMPLE}
Next, we show the proposed ETL framework in a numerical simulation\footnote{The Matlab code will be made publicly available upon publication.}.
In particular, we demonstrate the following properties:
\begin{compactitem}
    \item The ETL framework detects changes in the dynamics;
    \item we can adapt to change and make the uncertainty ellipsoid of the estimator arbitrarily small during dedicated learning experiments; and
    \item due to the MPC-based nature of the proposed learning experiments, we can simultaneously shrink the uncertainty and satisfy state and control constraints.
\end{compactitem}
To contrast our approach to adaptive approaches, we compare our method to continuously updating model parameters.
We show that continuously updating methods return after the learning experiment back to a steady-state, where the estimator performs a random walk in an ellipsoid that can grow arbitrarily large. Our method does not suffer from this issue.

\begin{algorithm2e}[tb]
	\DontPrintSemicolon
	\KwData{Initial model\;
		Initial covariance parameters\;
		Set up parameter filter\; 
		Set up nominal controller\;
		Set mode = "control"}
	\While{true}{
		Measure state\;
		Update parameter filter estimate and uncertainty\;
		\Switch{mode}{
			\Case{"control"}{
				Evaluate trigger condition\;
				\If{trigger condition = true}{
					Set mode = "experiment"\;
				}
				Apply nominal input
			}
			\Case{"experiment"}{
				Check condition to stop experiment\;
				\If{stopping condition = true}{
					Set mode = "control"\;
					Update model and controller
				}
				Apply experiment input
			}
			
		}
	}
	\caption{Event-triggered learning algorithm\label{algo}}
\end{algorithm2e}

\subsection{Setup}
The considered system describes the dynamics of a servomechanism consisting of a DC-motor, a gear-box, an elastic shaft, and an uncertain load, which is adopted from~\cite{BEMPORAD1998} and~\cite{schwenkel2020online}.
We chose this example as it allows us to illustrate the efficacy of our approach to changes in the dynamics, which will be represented by changes in the load.  While there are also other control approaches to deal with load changes, we want to demonstrate that our method is flexible and can readily cope with a problem without any further adaptation.

The continuous-time state-space equations are given by
\begin{align}
	\Dot{x} = 
	\left[\begin{matrix}
		0 &1 &0 &0\\
		-\frac{k_\theta}{J_\mathrm{L}} &-\frac{\beta_\mathrm{L}}{J_\mathrm{L}} &\frac{k_\theta}{\rho J_\mathrm{L}} &0\\
		0 &0 &0 &1\\
		\frac{k_\theta}{\rho J_\mathrm{M}} &0 &-\frac{k_\theta}{\rho^2 J_\mathrm{M}} &-\frac{\beta_\mathrm{M} R+K_\mathrm{T}^2}{J_\mathrm{M} R}\\
	\end{matrix}\right]
	x + 
	\left[\begin{matrix}
		0\\
		0\\
		0\\
		\frac{K_\mathrm{T}}{R J_\mathrm{M}}\\
	\end{matrix}\right]
	u\,,
\end{align}
where the state vector $x\in\mathbb{R}^4$ consists of the load angle, the motor angle and their time-derivatives. 
The input $u$ corresponds to the input DC voltage in Volt, which is constrained by $|u| \le 220$.
In addition, the state constraint $\left|\begin{bmatrix} k_\theta &0 &\sfrac{-k_\theta}{\rho} &0 \end{bmatrix} x\right|\le 78.5398$ must be fulfilled.
In the example, all model parameters are known (and can be found in~\cite{BEMPORAD1998}) except for the load inertia $J_\mathrm{L}$, which has the uncertainty range $10\, J_\mathrm{M} \le J_\mathrm{L} \le 30\, J_\mathrm{M}$.
The continuous-time model is converted into a discrete-time model with the sampling time $T_\mathrm{s} = 0.1\,\mathrm{s}$ using zero-order hold on the input. For the sake of notational convenience, we omit the units for the discretized system. 
On the discrete-time system, additive disturbances $w_k$ are introduced, where $w_k$ is drawn from a normal distribution with zero mean and covariance matrix $\Sigma_{w} = \mathrm{diag}\left( \begin{bmatrix} 0.99 & 0.99 & 0.939 & 0.056 \end{bmatrix}\cdot10^{-4}\right)$. 

Simulations of $3000$ time steps are performed. During the first 1000 steps, the nominal system with $ J_\mathrm{L} = 20\, J_\mathrm{M}$ is used. At time steps $1000$ and $2000$, the simulated system changes its parameter to $ J_\mathrm{L} = 22\, J_\mathrm{M}$ and  $ J_\mathrm{L} = 19\, J_\mathrm{M}$, respectively.
A standard MPC with zero terminal constraint is introduced based on the nominal model to control the system. The prediction horizon is set to 6 time steps.
We use the Kalman filter, as described in Section~\ref{subsec:KF}, for the parameter filter. The assumed covariance of system changes $\Sigma_{z}$ is designed using prior knowledge about possible parameter variations.
Thus, parameters that are unlikely to change are estimated very precisely, while the estimation of possibly changing parameters is adapting more rapidly.

To compare our ETL approach to nominal control without model adaption and to permanently adapted models, three similar simulations are performed.

\subsection{Event-triggered learning}
First, the ETL approach is applied to the system. The model and the controller is kept constant if no change is detected.
For change detection, the test statistic from Theorem~\ref{thm:KFTrigger} is computed. The graph of the test statistic over time is depicted in Figure~\ref{fig:tikz:timeETL}.
It is visible that the test statistic rises when the system parameters are changed after time step $1000$ and $2000$.
Eventually, the threshold is exceeded, which means that a significant model error is detected.
Then, a learning experiment using the experiment MPC~\eqref{eq:expMPC} is performed. In this period of $200$ time steps, a higher excitation of the system is generated.
This fact can also be noticed in Figure~\ref{fig:tikz:trajectoryETL}, where the trajectory of angles is shown divided into nominal control and experiment operation.
After the experiment, the current parameter estimate is set as the new model and used to update the controller and the test.
Model updates only happen directly after the learning experiments. For the remaining time, the model is kept constant.

\begin{figure}[tb]
	\includestandalone[width=\columnwidth, mode=buildnew]{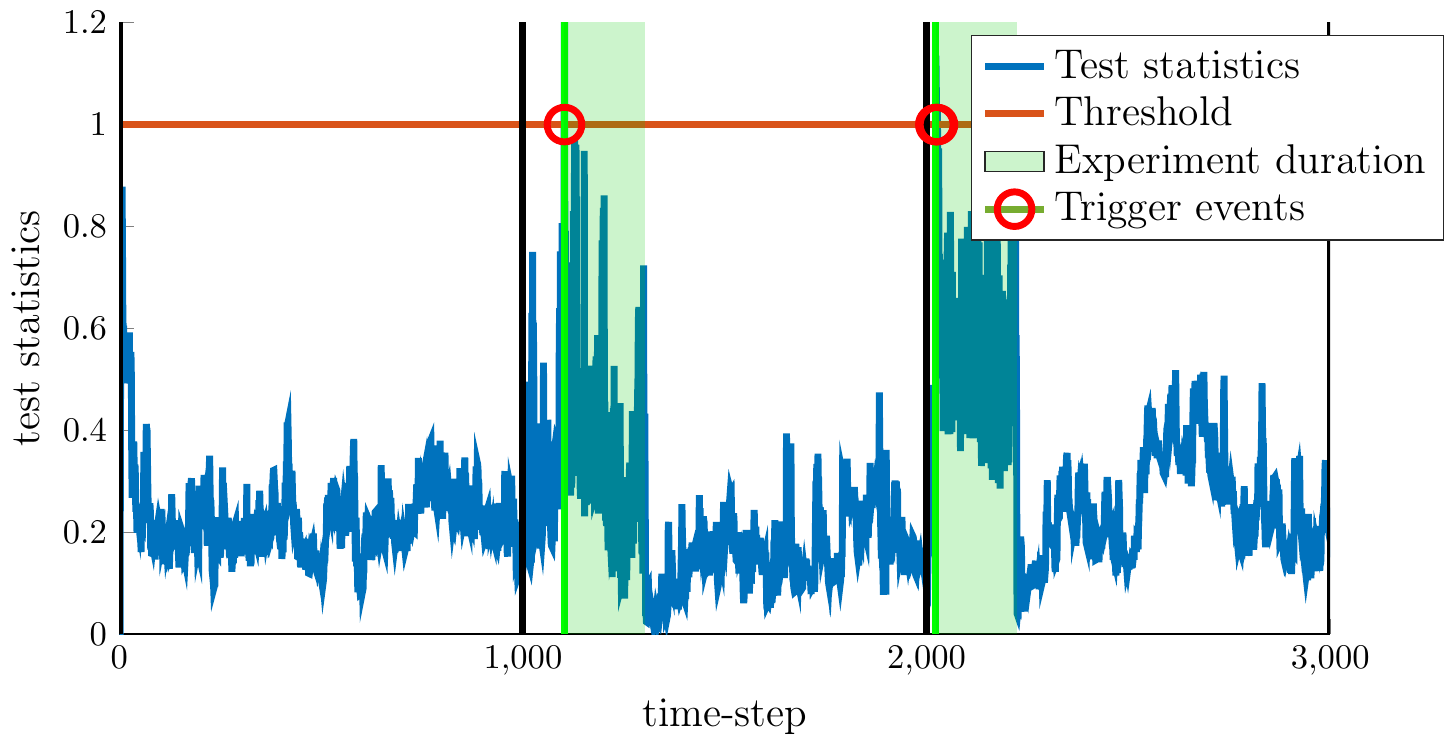}
	\caption{
		Test statistics of the learning trigger.
		The test condition \eqref{eq:KFTest} is evaluated in every time step. 
		In the figure, the test statistics and the threshold are normalized, such that a test statistic larger than one leads to a learning trigger event. This happens shortly after the true system parameters change at time steps 1000 and 2000, respectively.
		Thus, changing system parameters are detected reliably without triggering when the parameters stay constant.
	}
	\label{fig:tikz:timeETL}
\end{figure}

\subsection{Permanent model updates}\label{sec:alwaysUpdates}
In the second simulation, in every time step, the current parameter estimate is used to update the model and the controller. 
The estimation algorithm is capable of tracking the true system parameters. However, without significant excitation, the uncertainty stays large, and the estimated parameters perform a random walk inside the uncertainty ellipsoid.
This behavior is depicted in Figures~\ref{fig:tikz:paramETL} and~\ref{fig:tikz:paramAlways}, where the estimation error of four exemplary parameters is shown for the ETL and the permanent update approach.
In Figure~\ref{fig:tikz:paramAlways}, the current estimate is always used as the model (permanent updates), which is also used for control. Thus, the model is always varying, which can be seen based on the fluctuating error (red line).
In contrast, as depicted in Figure~\ref{fig:tikz:paramETL}, in the ETL approach, the model is kept constant unless a trigger occurs. Before updating to a new model, the system is excited in the experiment to obtain higher precision.
If the estimation error after the update grows, this does not affect the used model and the control algorithm.

\begin{figure}[tb]
	\includestandalone[width=\columnwidth, mode=buildnew]{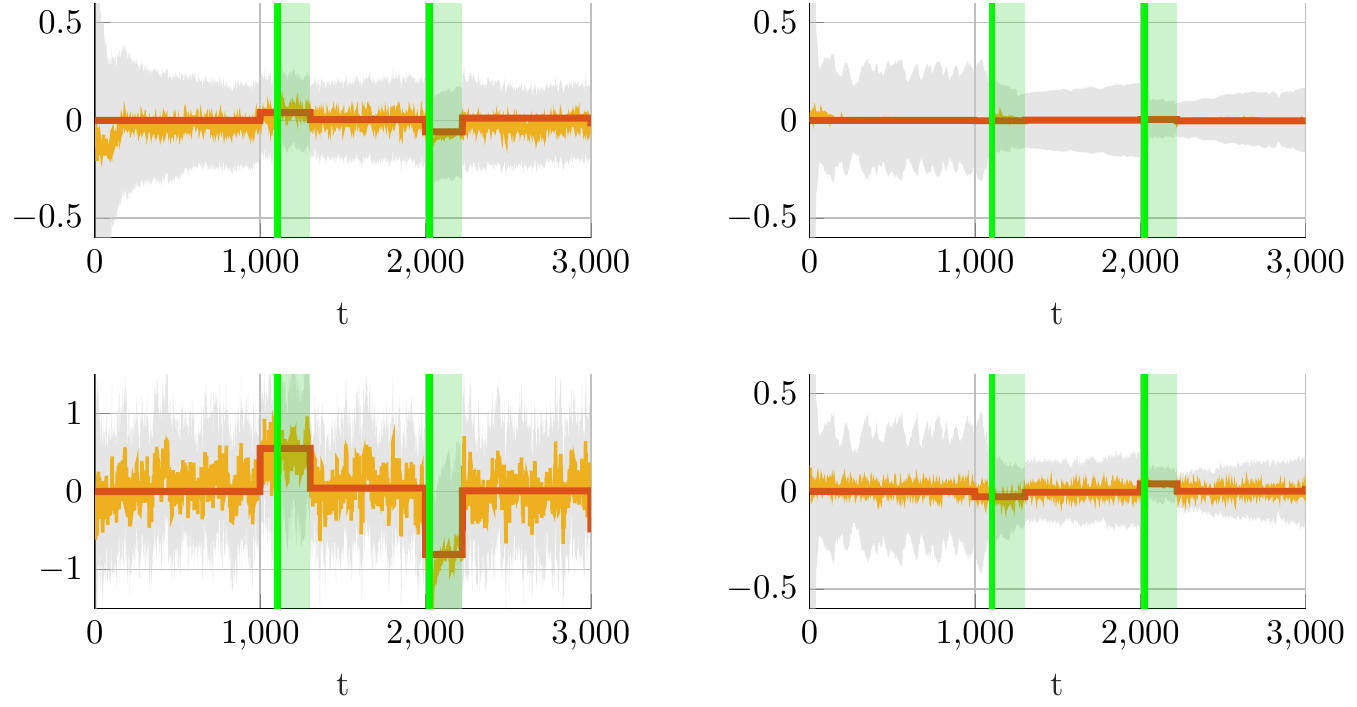}
	\caption{
		Evolution of four exemplary parameters and their estimate using ETL.
		In red, we show the error between system and model parameters. The yellow line displays the estimated model error using the current parameter point estimate. The gray region depicts a projection of the confidence ellipsoid onto these parameters centered around the estimated error.
		The model error stays close to zero, while the estimation error varies.
		When the gray region does not include zero, a learning experiment is triggered that excites the system to force the yellow line toward the red one.
		Thus, the uncertainty is reduced, and after the experiment, the model is updated with a precise estimate.
	}
	\label{fig:tikz:paramETL}
\end{figure}
\begin{figure}[tb]
	\includestandalone[width=\columnwidth, mode=buildnew]{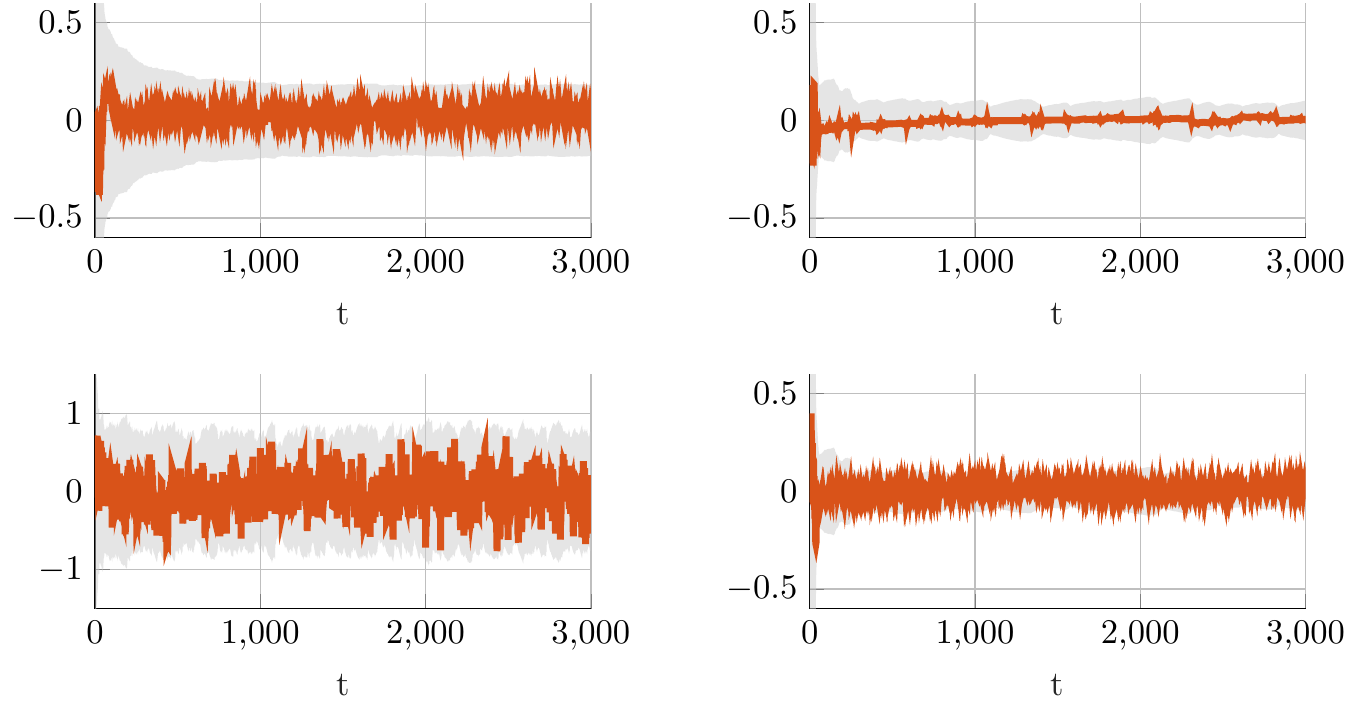}
	\caption{
		Evolution of the same four exemplary parameters and their corresponding estimates using permanent model updates.
		In contrast to Figure~\ref{fig:tikz:paramETL}, the model error oscillates heavily. It is not even possible to detect the introduced changes in the system dynamics (at $t=1000$ and $t=2000$) within this random walk.
	}
	\label{fig:tikz:paramAlways}
\end{figure}

\subsection{No model updates}\label{sec:noUpdates}
Third, the same kind of simulation is performed without adapting the model and the controller.
This corresponds to purely relying on the robustness of the control system against small parameter deviations.
In Figure~\ref{fig:tikz:paramNever}, the actual model error for four selected parameters is shown. 
As no update is carried out, the error will remain large when the system changes its parameters.
This can be disadvantageous as this error deteriorates the control performance.
In Figure~\ref{fig:tikz:trajectoryAlwaysNever}, it is visible that the state constraints are slightly violated in the experiments using permanent updates and constant nominal models that are never adopted, respectively.

\begin{figure}[tb]
	\includestandalone[width=\columnwidth, mode=buildnew]{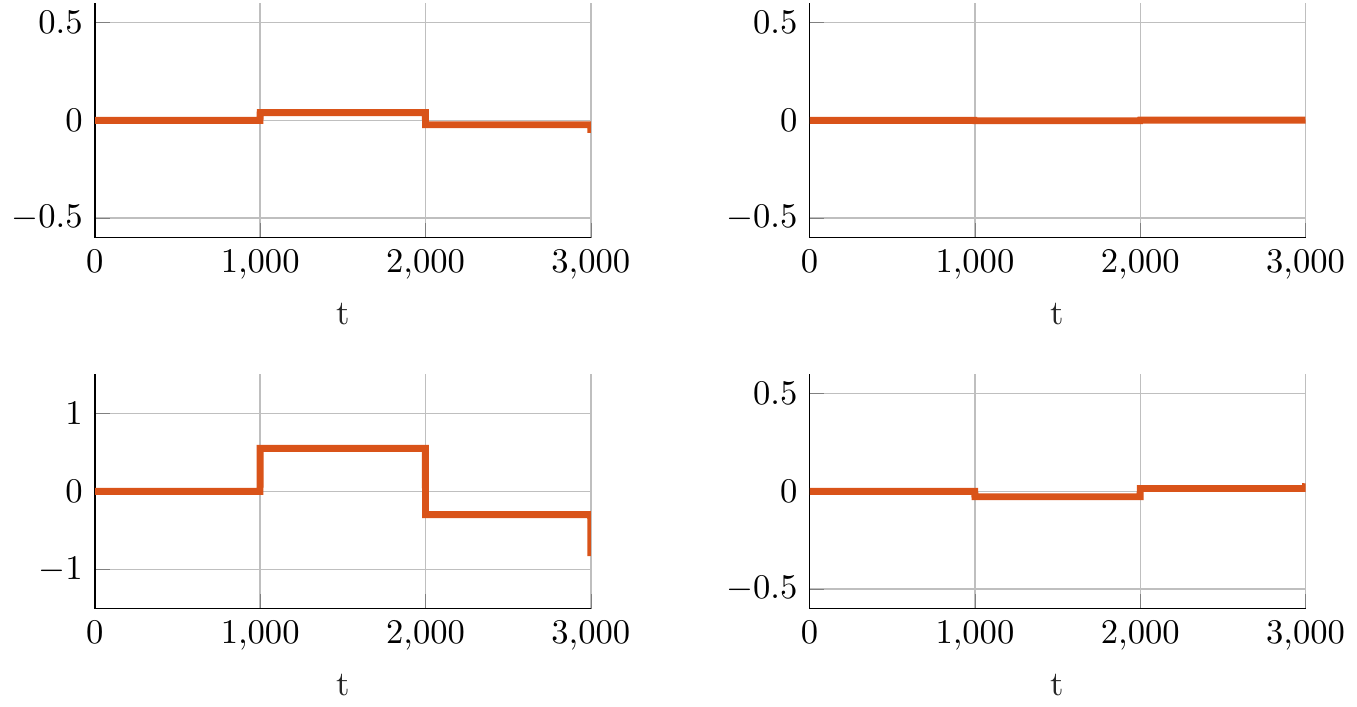}
	\caption{
		Evolution of four exemplary parameters and their estimate using the nominal model for control.
		Without model updates, the model contains errors when the system changes its parameters.
	}
	\label{fig:tikz:paramNever}
\end{figure}

\begin{figure}[tb]
	\includestandalone[width=\columnwidth, mode=buildnew]{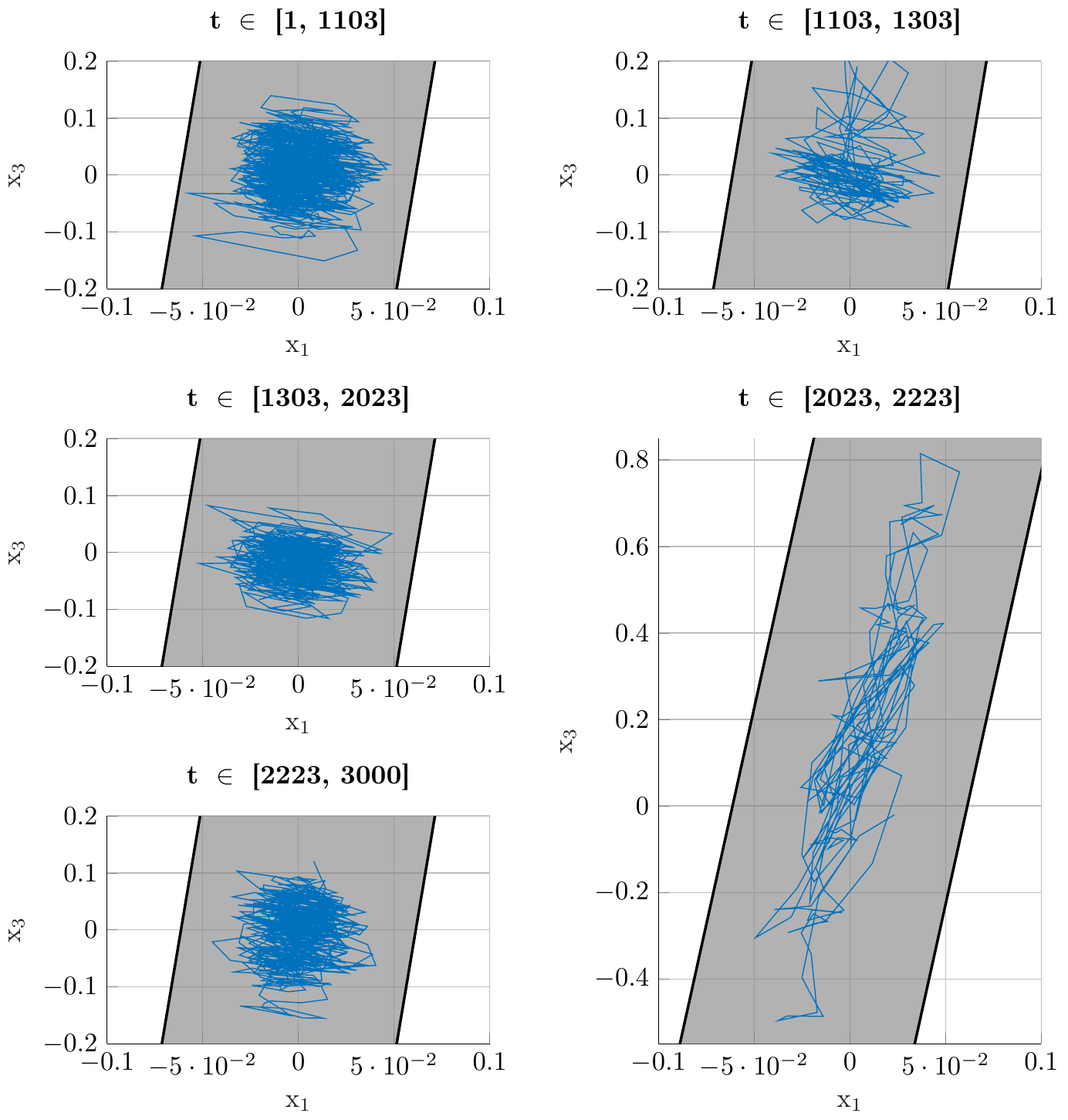}
	\caption{
		Trajectory of the angles $x_1$ and $x_3$ using ETL for control partitioned by simulation section. During the experiments (right column), the excitation is clearly larger.
		The gray area depicts the state constraints.
	}
	\label{fig:tikz:trajectoryETL}
\end{figure}

\begin{figure}[tb]
	\begin{subfigure}[b]{0.48\columnwidth}
		\centering
		\includestandalone[width=\textwidth, mode=buildnew]{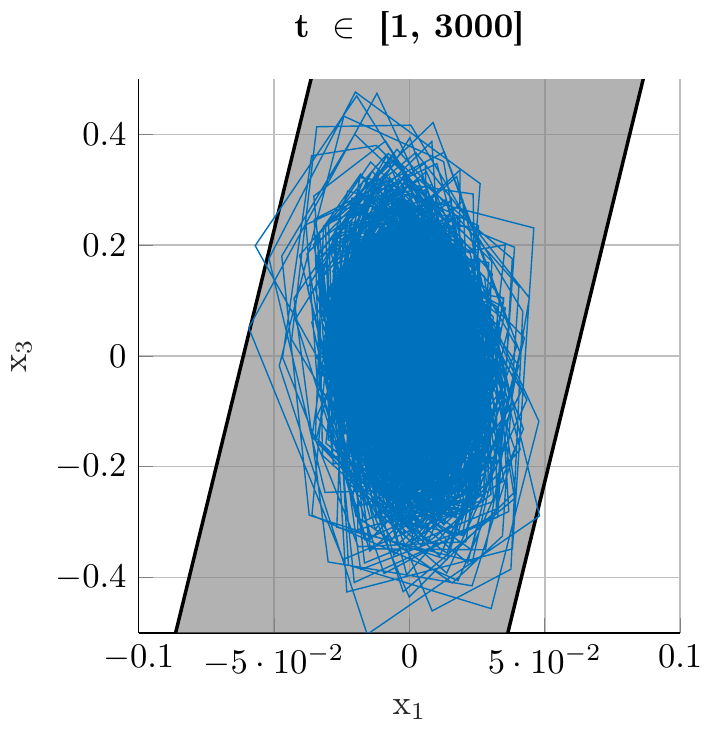}
		\caption{Permanent model updates}
		\label{fig:tikz:trajectoryAlways}
	\end{subfigure}
	\hfill
	\begin{subfigure}[b]{0.48\columnwidth}
		\centering
		\includestandalone[width=\textwidth, mode=buildnew]{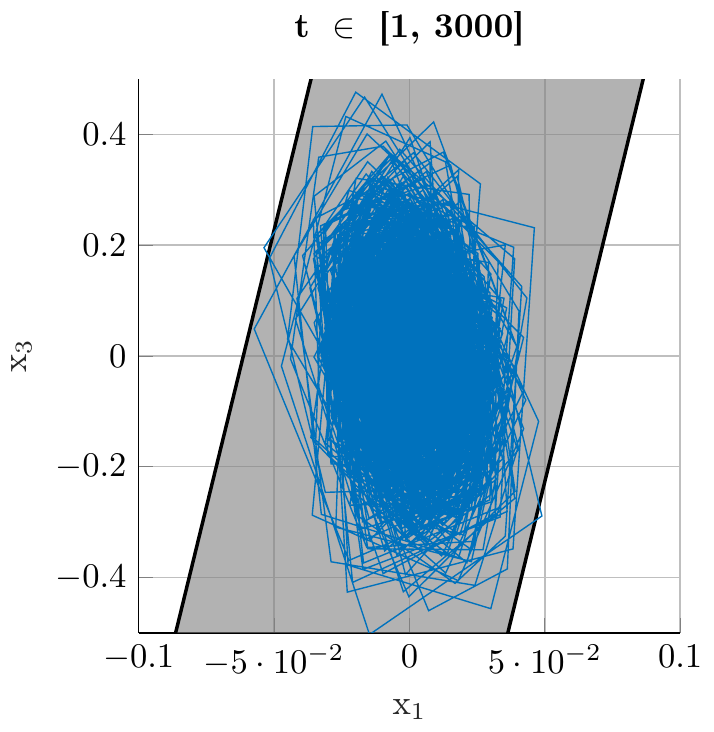}
		\caption{Nominal model}
		\label{fig:tikz:trajectoryNever}
	\end{subfigure}
	\caption{
		Trajectory of the angles $x_1$ and $x_3$ using (\ref{fig:tikz:trajectoryAlways}) permanent model updates for control and (\ref{fig:tikz:trajectoryNever}) the nominal model for control.
		The color-coding is the same as in Figure~\ref{fig:tikz:trajectoryETL}. In contrast to the ETL approach in Figure~\ref{fig:tikz:trajectoryETL}, we can observe state constraint violations here.
	}
	\label{fig:tikz:trajectoryAlwaysNever}
\end{figure}

\subsection{Results}
Updating the model \emph{all the time} or \emph{never} are both problematic. 
Without any updates (Section \ref{sec:noUpdates}), the nominal model will never adapt to changes. Thus, the model error will indefinitely cause problems.
On the other hand, updating the model in every time step (Section \ref{sec:alwaysUpdates}) with uninformative data also results in poor models even when the true system remains unchanged. Both of these extremes are problematic and lead to issues in terms of violating constraints (\cf Figure~\ref{fig:tikz:trajectoryAlwaysNever}) and performance. 

Event-triggered learning addresses these issues and keeps the model constant as long as no significant deviations are detected between the model and the current estimate.
Thus, a change of the model is only performed if the data indicates a significant model error.
Keeping the model fixed during nominal operation prevents the illustrated divergence issues of permanent updates.

The effect of the learning experiment on the estimation error is shown by the results in Table~\ref{tab:averageError}, where the average squared model parameter error is given for the different simulations.
If we neglect the periods of the simulation in which learning experiments took place ("Excluding excitation" in Table \ref{tab:averageError}), the model error is by far the smallest for ETL.
This once again highlights that the excitation during the learning experiment yields an accurate model.

If the whole duration of the simulation is taken into account, the permanently updated model has a smaller average error than the simulation with ETL and learning experiments.
This is a result of the immediate response of the estimate -- and hence also of the model and controller --  after a system change.
In ETL, we still considered the old model belief, which we have detected as inaccurate, during the learning experiment. 
However, obtaining a new, accurate model takes some time and requires excitation, and thus, learning has its own price.
Clearly, the longer the model remains unchanged after learning, the higher the benefit of an accurate model.  

\begin{table}[tb] %
	\caption{Average squared model parameter error [\SI{e-3}{}].}
	\label{tab:averageError}
	\centering
	\begin{tabular}{ c | c  c c }%
		\toprule
		&ETL& Permanent updates&No updates\\
		\midrule
		Whole simulation&4.359&3.197&7.123\\
		Excluding excitation&1.065&3.322&6.581\\
		\bottomrule
	\end{tabular}
\end{table}

\section{CONCLUSION AND FUTURE WORK}

We propose a parameter filter-based learning trigger that generalizes previous work on event-triggered learning. 
By considering the inherent notion of uncertainty provided by Kalman-type filters, we can combine point estimates with powerful statistical tests to trigger learning experiments on necessity.
Further, we are able to provide optimal excitation signals that specifically target parameters with high uncertainty. 

We provide an ETL framework that joins i) learning trigger, ii) online model learning, and iii) experiment design. The core assumption that makes this possible is the linearity of the dynamics, which is rooted deep inside the parameter filter. Extending this to nonlinear systems, which we plan for future work, requires significant extensions of all aspects i)--iii). First steps for the trigger design, we have taken in~\cite{solowjow2020kernel}, and a new type of causal model learning was proposed in~\cite{baumann2020identifying}. Further, Gaussian process-based approaches such as~\cite{umlauft2020smart} might also be a promising approach for nonlinear extensions.





 \printbibliography

\begin{IEEEbiography}[{\includegraphics[width=1in,height=1.25in,clip,keepaspectratio]{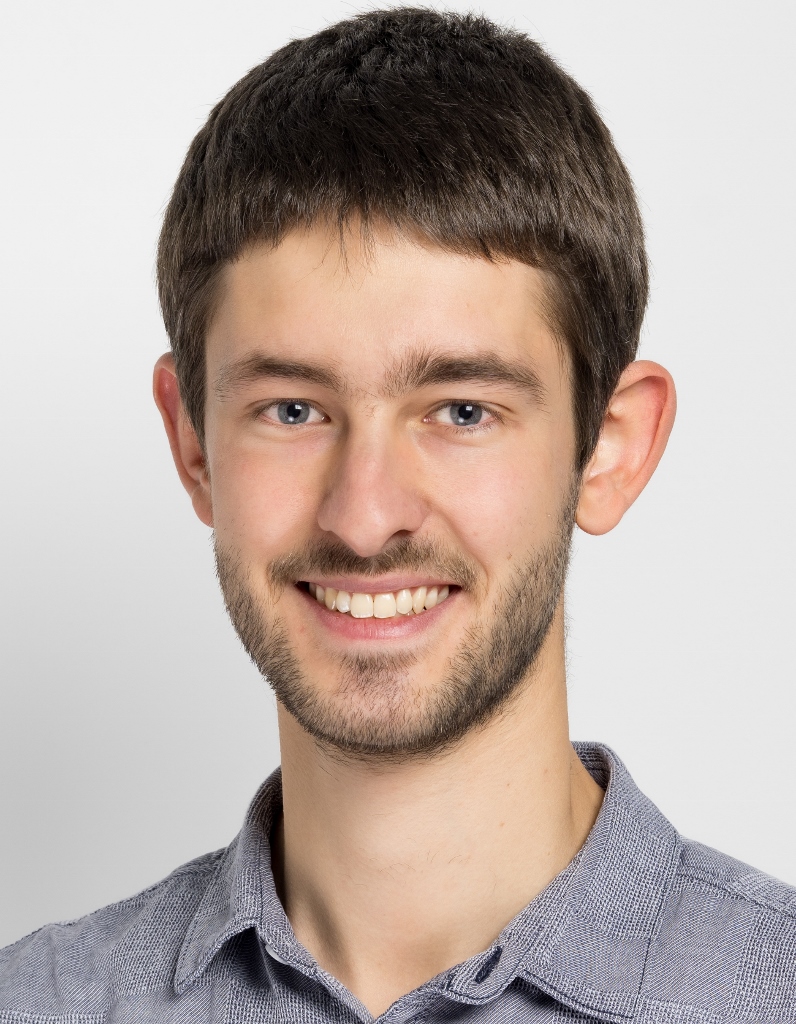}
	}]{Sebastian Schlor} received his B.\ Sc.\ and M.\ Sc.\ degree in Engineering Cybernetics from the University of Stuttgart, Stuttgart, Germany, in 2017 and 2020, respectively, and is currently a Ph.\ D.\ student at the University of Stuttgart with the Institute for Systems Theory and Automatic Control under the supervision of Prof.\ Frank Allgöwer.
His research interests include the area of privacy and security of dynamical systems and encrypted control.
\end{IEEEbiography}

\begin{IEEEbiography}[{\includegraphics[width=1in,height=1.25in,clip,keepaspectratio]{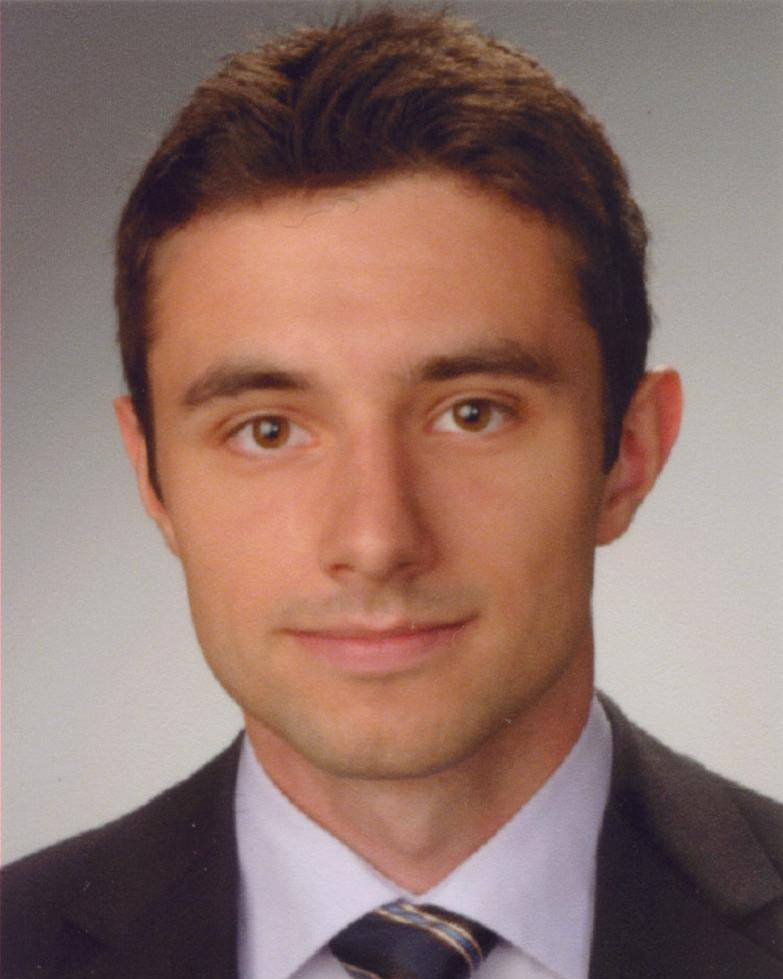}
	}]{Friedrich Solowjow} received B.\ Sc.\ degrees in
Mathematics and Economics from the University
of Bonn in 2014 and 2015, respectively, and a M.\ Sc.\
degree in Mathematics also from the University of
Bonn in 2017. He is currently a Ph.\ D.\ student in
the Intelligent Control Systems Group at the Max
Planck Institute for Intelligent Systems, Stuttgart,
Germany and a member of the International Max
Planck Research School for Intelligent Systems. His
main research interests are in systems and control
theory and machine learning.
\end{IEEEbiography}

\begin{IEEEbiography}[{\includegraphics[width=1in,height=1.25in,clip,keepaspectratio]{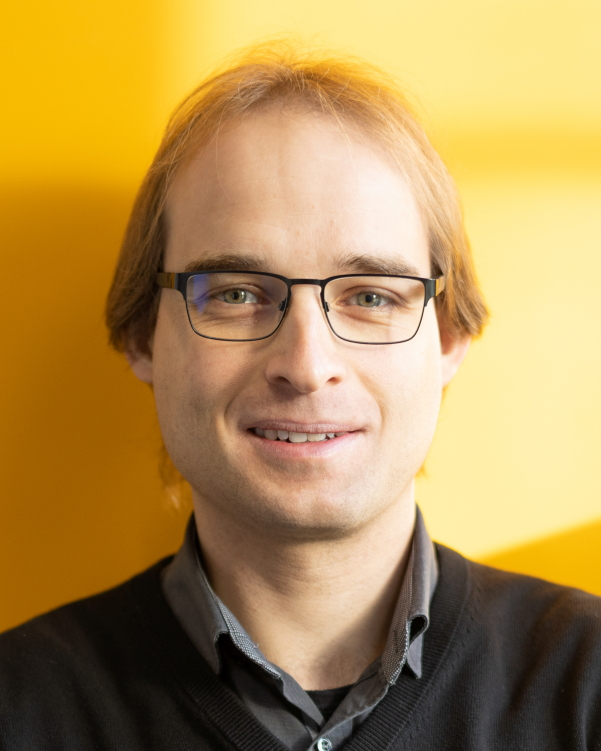}
	}]{Sebastian Trimpe} (M’12) received the B.\ Sc.\ degree
in general engineering science and the M.\ Sc.\ degree
(Dipl.-Ing.) in electrical engineering from Hamburg
University of Technology, Hamburg, Germany, in
2005 and 2007, respectively, and the Ph.\ D.\ degree
(Dr.\ sc.) in mechanical engineering from ETH Zurich,
Zurich, Switzerland, in 2013. Since 2020, he has been a
full professor at RWTH Aachen University, Germany,
where he heads the Institute for Data Science in Mechanical Engineering. Before, he was an independent
Research Group Leader at the Max Planck Institute
for Intelligent Systems in Stuttgart and Tübingen, Germany. His main research
interests are in systems and control theory, machine learning, networked and
autonomous systems. Dr. Trimpe is the recipient of several awards, among others,
the triennial IFAC World Congress Interactive Paper Prize (2011), the Klaus
Tschira Award for achievements in public understanding of science (2014),
and the Best Paper Award of the International Conference on Cyber-Physical
Systems (2019).
\end{IEEEbiography}

\end{document}